\documentclass{siamart0216}

\usepackage{amsfonts}
\usepackage{mathtools}
\usepackage{graphicx}
\usepackage[caption=false]{subfig}
\usepackage{epstopdf}
\usepackage{algorithmic}
\usepackage{multirow}
\usepackage{booktabs}
\usepackage{tabularx}
\usepackage{pgfplots}
\usepgfplotslibrary{groupplots}
\pgfplotsset{width=7cm,compat=newest}
\usepackage{tikz}
\usetikzlibrary{calc, shapes, positioning, decorations.pathreplacing}
\ifpdf
  \DeclareGraphicsExtensions{.eps,.pdf,.png,.jpg}
\else
  \DeclareGraphicsExtensions{.eps}
\fi
\hypersetup{pdfstartview={FitH}}
\pgfplotsset{
	cycle list={
		{red,mark=*},
		{blue,mark=square*},
		{orange,mark=triangle*},
		{gray,mark=diamond*},
	},
}
\newcommand{\TheTitle}{AMPS: An Augmented Matrix Formulation for Principal Submatrix
Updates with Application to Power Grids} 
\newcommand{\TheAuthors}{Y.-H. Yeung, A. Pothen, M. Halappanavar, and Z. Huang}

\headers{Augmented Matrix Formulation for Principal Submatrix Updates}{\TheAuthors}

\title{{\TheTitle}\thanks{This work was supported in part by 
the National Science Foundation grant CCF-1552323
and by the Applied Mathematics Program within the Office of Science of the 
U.S. Department of Energy by grant DE-SC0010205 at Purdue, and 
at the Pacific Northwest National Laboratory, 
operated by Battelle for the  DOE under Contract DE-AC05-76RL01830. 
}}

\author{
  Yu-Hong Yeung\thanks{Department of Computer Science, Purdue University, West Lafayette, IN
		(\email{yyeung@purdue.edu}, \email{apothen@purdue.edu}).}
  \and
	Alex Pothen\footnotemark[2]
	\and
	Mahantesh Halappanavar\thanks{Pacific Northwest National Laboratory, Richland, WA
	(\email{hala@pnnl.gov}, \email{Zhenyu.Huang@pnnl.gov}).} 
	\and 
	Zhenyu Huang\footnotemark[3]
}

\usepackage{amsopn}

\newcommand\closure[2]{\mathop{\mathrm{closure}_{#1}\left(#2\right)}}
\newcommand\struct[1]{\mathop{\mathrm{struct}\left(#1\right)}}
\newcommand\vb{\boldsymbol}

\ifpdf
\hypersetup{
  pdftitle={\TheTitle},
  pdfauthor={\TheAuthors}
}
\fi

\begin{document}

\maketitle

\begin{abstract}
We present AMPS, an augmented matrix approach to update the solution
to a linear system of equations when the matrix is modified 
by a few elements within a principal submatrix. 
This problem arises in the dynamic security  analysis of a power grid, 
where operators need to perform $N-k$ contingency analysis, 
i.e., determine the state of the system when exactly $k$ links from $N$ fail.  
Our algorithms augment the matrix to account for the changes 
in it, and then  compute the solution to the
augmented system without refactoring the modified matrix. 
We provide two algorithms, a direct method,
and a hybrid direct-iterative method for solving the augmented system. 
We also exploit the sparsity of the matrices and vectors to accelerate the
overall computation. 
We analyze the time complexity of both algorithms, and show 
that it is bounded by the number of nonzeros in a subset of the columns of the
Cholesky factor that are selected by the nonzeros in the sparse right-hand-side vector.  
Our algorithms are compared on three power grids with PARDISO, 
a parallel direct solver, and CHOLMOD, a direct solver with the 
ability to modify the Cholesky factors of the matrix. 
We show that our augmented algorithms outperform  PARDISO 
(by two orders of magnitude), and CHOLMOD (by a factor of up to 5). 
Further, our algorithms scale better 
than CHOLMOD as the number of elements updated increases. 
The solutions are computed with high accuracy. 
Our algorithms are capable of computing $N-k$ contingency analysis on 
a 778 thousand bus grid, updating a solution with $k=20$ elements 
in $16$ milliseconds on an Intel Xeon processor.
\end{abstract}

\begin{keywords}
  direct methods, iterative methods,
	augmented matrix,  sparse matrices, 
	matrix updates, powerflow analysis, contingency analysis
\end{keywords}

\begin{AMS}
   65F50, 65F10, 65F05, 65Y20
\end{AMS}

\section{Introduction}
\label{sec:introduction}
We consider updating the solution to a system of equations $ A\vb{x} = \vb{b}$,
where $A$ is a symmetric positive definite or indefinite $n \times n$ matrix
and $\vb{b}$ is an $n$-vector, when a low-rank change is made to $A$. 
The change we consider is an update of a principal submatrix of the form
\begin{equation}
	\hat{A} = A - H E H^\top,\label{eq:update}
\end{equation}
where $E$ is a symmetric $m \times m$ matrix, and $H$ is an $n \times m$
submatrix of an identity matrix, and $m \ll n$.
Since both $A$ and $E$ are symmetric, $\hat{A}$ is also symmetric. Note
that the dimension of the matrix does not change when it is updated.
We wish to compute the solution to the updated system 
\begin{equation}
	\hat{A}\hat{\vb{x}} = \hat{\vb{b}}.\label{eq:mod}
\end{equation}
We describe an augmented matrix approach to the solution of the updated system,
in which the augmented matrix is a block $3 \times 3$ matrix whose
$(1,1)$-block is the original matrix $A$, and the updates to $A$
are represented by the other submatrices of the block matrix. We describe two
algorithms to solve this augmented system. In both algorithms,
the original matrix $A$ is factored with a direct method. 
In the first algorithm, the Schur complement system is also solved by 
a direct method, and in the second algorithm it is solved with a 
Krylov subspace solver.  
We maintain symmetry in the augmented system of equations and the second
algorithm whereas in the first algorithm an unsymmetric system is solved 
to reduce the computation time.
Note that our algorithms can handle arbitrary changes to $\hat{\vb{b}}$
in \cref{eq:mod}. However, in the power grid application considered here,
$\hat{\vb{b}}$ only changes in the set of $m$ rows where the principal submatrix
is updated. Hence we focus on this situation in our experiments.

Our motivation for this work comes from dynamically assessing the 
security of power grids, which is also  called contingency analysis.
In power engineering, an interconnected power system is described by a system 
of complex, nonlinear equations representing the relationship between 
voltages, powers and admittances at points of interest called buses. 
Here, we consider the ``DC'' approximation of this problem, 
which is derived using heuristic assumptions, and is 
described by a linear system,
\begin{equation}
	-B\vb{d} = \vb{p},\label{eq:DC}
\end{equation}
where $B$ is the imaginary component of the $n \times n$ admittance matrix, 
$\vb{d}$ is an $n$-vector of the relative phase shift of the voltage,  
$\vb{p}$ is an $n$-vector of the real power, and  
$n$ is the number of buses in the system. 
In contingency analysis, one removes an existing connection 
between two buses in the system to simulate the failure to transmit power 
through that transmission line, or all the connections to a generator to 
simulate the failure to generate power from it. 
Removing a connection in the system corresponds to a principal submatrix 
update to \cref{eq:DC},  
and the updated matrix has the same size $n$ as the original matrix.
Bienstock discusses a mixed-integer programming approach to this problem
\cite{bienstock}, which restricts the size of the problems they can 
solve to a few hundred buses.   

We propose AMPS, an augmented system that is equivalent to \cref{eq:mod}, 
which means in exact arithmetic solving the augmented system would give us 
the same solution vector $\hat{\vb{x}}$. 
Our experimental  results show that the accuracy of the solution to 
the augmented system is comparable to  that of the solution $\hat{\vb{x}}$ 
obtained by solving \cref{eq:mod} by a direct method.

Our algorithm satisfies the following four desiderata:
\begin{enumerate}
	\item The solution of the augmented system should be computed in a number of
	 operations proportional to the size of the update $m$ rather than the size of
	 the system $n$. This is especially important for large systems when there is
	 a need for a sequence of updates in real-time.
	\item The accuracy of the solution to the augmented system should be
	 comparable to that of the direct solution of the modified system.
	\item Both the factors of the matrix and the solution of the original system
	 should be utilized in solving the augmented system to avoid redundant 
	 computations.
	\item Sparsity in the matrices and the vectors should be exploited to
	 accelerate the computations. 
\end{enumerate}

The work most closely related to this paper is an augmented matrix approach to
solving the stiffness system of equations in a surgery stimulation when an organ
is cut or deformed, proposed by Yeung, Crouch and Pothen~\cite{yeung}.
The surgery is visualized by updating a finite element formulation of a linear
elastic model of the organ as it is cut. The matrix here is the (varying)
stiffness matrix from the finite element model of the organ. For surgery
simulations, solutions of tens or hundreds of modified systems per
second are needed. With the augmented matrix approach, the stiffness
matrix of the initial mesh can be kept unchanged, and all changes as the mesh
is being cut can be described using the $(1,2)$- and $(2,1)$-blocks of a 
block $2 \times 2$ matrix.
In this problem, nodes and elements could be deleted, added, or replaced,
and thus the dimension of the matrix changes, unlike the situation here. 
These authors used an unsymmetric form of the augmented matrix with a hybrid
direct-iterative algorithm, where a direct method was used to factor
the initial stiffness matrix, and the Schur complement system was solved
implicitly using a Krylov space solver. There are two major differences here. 
The first is that the update is restricted to a principal submatrix in the
power grid context. The second is that symmetry is preserved while it was
destroyed in the earlier method even though both the matrix and the updates
were symmetric. There are other existing augmented matrix approaches, which
will be discussed later in this paper.

\paragraph{Notation}
We use Householder notation throughout; that is, matrices are denoted by upper
case Roman letters, vectors by lower case Roman letters, and scalars by Greek
letters. There are some exceptions: Indices and dimensions are also denoted by
lower case Roman letters (e.g. $i$, $j$, $k$ and $m$, $n$). With this
convention, the elements of a matrix $A$ are denoted by $\alpha_{ij}$, and the
elements of a vector $\vb{x}$ are denoted by  $\chi_j$. A submatrix of $A$ is
denoted by $A_{ij}$, and a subvector of $\vb{x}$ is denoted by $\vb{x}_j$. 
We use $A^\top$ to denote the transpose of $A$. The symbols $L$ and $D$ are
reserved for lower triangular and (block) diagonal matrices. 
The $j$th column of the identity matrix $I$ is written as $\vb{e}_j$,
and thus the matrix $H$ in \cref{eq:update} is
$H = \left[\vb{e}_{j_1},\vb{e}_{j_2},\ldots,\vb{e}_{j_m}\right]$ for the set
of indices of the modified rows and columns
$\mathbb{S} = \{j_1, j_2, \ldots, j_m\}$.

\paragraph{Organization of this article}
\Cref{sec:formulation} presents our new augmented system of equations for
solving the modified system when a principal submatrix is updated. 
\Cref{sec:method} describes the details of the algorithm to solve the modified 
system using the augmented formulation presented in the previous section.
\Cref{sec:results} presents computational times and the accuracy 
of solutions when the augmented system is applied to contingency analysis of
three power grids. 
\Cref{sec:conclusions} discusses conclusions and directions for future work.

\section{Augmented system formulation}
\label{sec:formulation}
It is well known that augmented systems can be used to effectively add
and remove rows and columns of matrices~\cite{bisschop1977,gill}. 
We begin by describing how these operations are accomplished, 
assuming that both the original matrix and the modifications are symmetric, 
i.e., the procedures are applied to rows and columns simultaneously. 
These modifications are not restricted to principal submatrix updates.
Also, these modifications might not preserve the nonsingularity 
of the matrix. 
Hence after each update, we characterize the conditions that 
must be satisfied for the updated matrix to be nonsingular
when the initial matrix is nonsingular. 
These results are obtained using the determinantal identity 
\[
	\det{\begin{bmatrix} A&B \\ C&D \end{bmatrix}} = \det{(A)} \det{(D - CA^{-1} B)},
\]
when $A$ is nonsingular. 
The goal of these characterizations is to show that our 
augmented system formulation
by itself does not create singular matrices.

\subsection{Adding a row and a column}
To add a row and a column to $A\vb{x} = \vb{b}$, we consider the system
\begin{equation}
	\begin{bmatrix}
		A & \vb{\hat{a}}\\
		\vb{\hat{a}}^\top & \hat{\alpha}
	\end{bmatrix} \begin{bmatrix}
		\vb{\hat{x}}_1 \\ \hat{\chi}
	\end{bmatrix} = \begin{bmatrix}
		\vb{b} \\ \hat{\beta}
	\end{bmatrix}.
\end{equation}
If $A$ is nonsingular and 
$\hat{\alpha} \neq \vb{\hat{a}}^\top A^{-1} \vb{\hat{a}}$, 
then the augmented matrix is nonsingular;  
and if $A$ is positive definite and 
$\hat{\alpha} > \vb{\hat{a}}^\top A^{-1} \vb{\hat{a}}$, 
then the augmented matrix is also positive definite. 

\subsection{Removing a row and a column}
To remove the $j$th row and column from $A\vb{x} = \vb{b}$,
we consider the system
\begin{equation}
	\begin{bmatrix}
		A & \vb{e}_j\\
		\vb{e}_j^\top & 0
	\end{bmatrix} \begin{bmatrix}
		\vb{\hat{x}}_1 \\ \hat{\chi}
	\end{bmatrix} = \begin{bmatrix}
		\vb{b} \\ 0
	\end{bmatrix}.\label{eq:remove}
\end{equation}
The last row $\vb{e}_j^\top \vb{\hat{x}}_1 = 0$ constrains the $j$th
component of $\vb{\hat{x}}_1$ to be $0$, and consequently removes the
contribution of the $j$th column of $A$. This leaves us with one fewer
\textit{effective} variable than the number of equations. This is compensated
by the additional component $\hat{\chi}$ in the solution vector. 
Consider the $j$th row of the augmented system: 
$\vb{e}_j^\top A \vb{\hat{x}}_1 + \hat{\chi} = \vb{e}_j^\top\vb{b}$.
Since $\hat{\chi}$ only appears in the $j^\text{th}$ row of the system, it is
constrained to the value $\vb{e}_j^\top(\vb{b} - A\vb{\hat{x}}_1)$
after all the other components of $\vb{\hat{x}}_1$ are determined. 
Its value will be discarded after the system is solved.

Augmentation in this manner makes the matrix symmetric indefinite. If $A$ is a
symmetric positive definite matrix, then we can show that the augmented matrix
is nonsingular, since its determinant is equal to $-\det(A)\, (A^{-1})_{jj}$,
and both terms are positive. 

\subsection{Replacing a row and a column}
Replacing a row and a column can be done by removing the old row and column and
adding the new ones. The resulting augmented formulation would be
\begin{equation}
	\begin{bmatrix}
		A & \vb{\hat{a}}_j & \vb{e}_j\\
		\vb{\hat{a}}_j^\top & \hat{\alpha}_{jj} & 0\\
		 \vb{e}_j^\top & 0 & 0
	\end{bmatrix} \begin{bmatrix}
		\vb{\hat{x}}_1 \\ \hat{\chi}_1 \\ \hat{\chi}_2
	\end{bmatrix} = \begin{bmatrix}
		\vb{b} \\ \hat{\beta} \\ 0
	\end{bmatrix},
\end{equation}
where $\vb{\hat{a}}_j$ and $\hat{\alpha}_{jj}$ are the $j$th column and the
$(j,j)$th element of $\hat{A}$ in \cref{eq:mod} respectively. Note that the
$j$th component of $\vb{\hat{a}}_j$ is then multiplied by the
$j$th component of $\vb{\hat{x}}_1$ which is constrained to be $0$ 
by the last equation. 
Hence the $j$th component of $\vb{\hat{a}}_j$ can be chosen arbitrarily.

We can calculate the determinant of the 
augmented matrix as 
\[
	\det{(A)} \ \det{\begin{bmatrix}
			\hat{\alpha}_{jj} - \vb{\hat{a}}_j^\top A^{-1} \vb{\hat{a}}_j &
			-\vb{\hat{a}}_j^\top A^{-1} \vb{e}_j \\ 
			-\vb{e}_j^\top A^{-1} \vb{\hat{a}}_j  &
			-\vb{e}_j^\top A^{-1} \vb{e}_j 
	\end{bmatrix}}.
\]
Hence if $A$ and the $2 \times 2$ matrix above are both nonsingular,
the augmented matrix is also nonsingular.

\subsection{Replacing multiple rows and columns}
Replacing $m$ rows and\linebreak columns can be done by concatenating the
replaced rows and columns. Suppose the set of indices of the rows and columns
to be replaced is $\mathbb{S} = \{j_1, j_2, \ldots, j_m\}$. The complete
augmented formulation would be
\begin{equation}
	\begin{bmatrix}
		A & J & H\\
		J^\top & C & 0\\
		H^\top & 0 & 0
	\end{bmatrix} \begin{bmatrix}
		\vb{\hat{x}}_1 \\ \vb{\hat{x}}_2 \\ \vb{\hat{x}}_3
	\end{bmatrix} = \begin{bmatrix}
		\vb{b} \\ H^\top\vb{\hat{b}} \\ \vb{0}
	\end{bmatrix},
	\label{eq:replace}
\end{equation}
where $J = \left[\vb{\hat{a}}_{j_1}, \vb{\hat{a}}_{j_2}, \ldots, \vb{\hat{a}}_{j_m}\right]$
are the modified columns of $\hat{A}$,
$H = \left[\vb{e}_{j_1}, \vb{e}_{j_2}, \ldots, \vb{e}_{j_m}\right]$ is the
submatrix of the identity matrix with the indices of the columns to be
replaced, and $C = H^\top \hat{A} H$ is the diagonal block of
the modified matrix $\hat{A}$ where the changes occur.
Note that the submatrix $C$ is $m \times m$, $J$ and $H$ are $n \times m$,
and $m \ll n$.

Again, if $A$ is nonsingular, we can express the determinant of the 
augmented matrix as 
\[	\det{(A)} \ \det{
\begin{bmatrix} C - J^\top A^{-1} J  & - J^\top A^{-1} H \\ 
- H^\top A^{-1} J & - H^\top A^{-1} H 
\end{bmatrix}}.
\]
If $A$ and the second matrix above are both nonsingular, then the augmented
matrix is also nonsingular. (We can choose $J = AH$ as shown later
in this section; then the block $2 \times 2$ matrix above is the negation of
the Schur complement matrix $S_1$ in the iterative variant of our AMPS
algorithm in \Cref{sec:method}.) 
In other words, the augmented matrix is
nonsingular if and only if both $A$ and the Schur complement matrix $S_1$ are
nonsingular.

We proceed to refine the system of equations \cref{eq:replace} further. 
With a suitable $n \times n$ permutation matrix $P$, we can partition $H$ into 
an identity matrix and a zero matrix:
\begin{equation}
	P H = \begin{bmatrix}
		I_m \\ 0_{n-m}
	\end{bmatrix}.\label{eq:PH}
\end{equation}
Applying the same permutation matrix $P$ to $J$, $A$, $\vb{\hat{x}}_1$ and
$\vb{b}$ yields
\begin{subequations}
\begin{align}
	P J &= \begin{bmatrix}
		J_1 \\ J_2
	\end{bmatrix},
	&P A P^\top &= \begin{bmatrix}
		A_{11} & A_{12}\\
		A_{12}^\top & A_{22}
	\end{bmatrix},\\
	P \vb{\hat{x}}_1 &= \begin{bmatrix}
		\vb{\hat{x}}_{11} \\ \vb{\hat{x}}_{12}
	\end{bmatrix},\label{eq:Px1}
	&P \vb{b} &= \begin{bmatrix}
		\vb{b}_1 \\ \vb{b}_2
	\end{bmatrix}.
\end{align}
\end{subequations}
We can then apply the permutation matrix
\begin{equation}
	\hat{P} = \begin{bmatrix}
		P \\ & I_m \\ & & I_m
	\end{bmatrix}
\end{equation}
to the matrix in \cref{eq:replace} from both left and right, which yields
\begin{equation}
	\begin{bmatrix}
		A_{11} & A_{12} & J_1 & I\\
		A_{12}^\top & A_{22} & J_2 & 0\\
		J_1^\top & J_2^\top & C & 0\\
		I & 0 & 0 & 0
	\end{bmatrix} \begin{bmatrix}
		\vb{\hat{x}}_{11} \\ \vb{\hat{x}}_{12} \\ \vb{\hat{x}}_2 \\ \vb{\hat{x}}_3
	\end{bmatrix} = \begin{bmatrix}
		\vb{b}_1 \\ \vb{b}_2 \\ H^\top\vb{\hat{b}} \\ \vb{0}
	\end{bmatrix}.
	\label{eq:principal}
\end{equation}
Here $A_{11}$ is the $m \times m$ submatrix being replaced by $C$, $A_{22}$ is
the $(n-m) \times (n-m)$ principal submatrix of $A$ that is unchanged, and
$A_{12}$ is the $m \times (n-m)$ off-diagonal submatrix of $A$. Note that the
third column block effectively replaces the first column block, and by symmetry
in the update, the third row block also replaces the first row block. Hence,
the submatrix $\displaystyle\begin{bmatrix}J_2\\C\end{bmatrix}$ must consist of
the modified columns in $\hat{A}$ that correspond to the original columns
$\displaystyle\begin{bmatrix}A_{12}^\top\\A_{11}\end{bmatrix}$ in $A$.

\begin{lemma}\label{lm:equiv}
	The submatrix $J_1$ in \cref{eq:principal} can be chosen arbitrarily such that
	the system is always consistent. Moreover, if $\hat{A}$ in \cref{eq:mod} is
	nonsingular, $\vb{\hat{x}}_{12}$ and $\vb{\hat{x}}_2$ are independent of 
	the submatrix $J_1$.
\end{lemma}

\begin{proof}
	Consider the last row block of \cref{eq:principal}. We have
	$\vb{\hat{x}}_{11} = \vb{0}$. Consequently, the first column block, which then
	multiplies $\vb{\hat{x}}_{11}$, does not contribute to the solution of the
	system. Moreover, consider the first row block of \cref{eq:principal}:
	\begin{equation}
		A_{12}\vb{\hat{x}}_{12} + J_1\vb{\hat{x}}_2 + \vb{\hat{x}}_3 = \vb{b}_1.
		\label{eq:coupled}
	\end{equation}
	Since $\vb{\hat{x}}_3$ only contributes to one row block in the system of
	equations, its values can be determined uniquely for any values of $J_1$.
	Hence the submatrix $J_1$ can be chosen arbitrarily. 

	Now we can prove the second statement in the lemma. If we consider the second
	and third row and column blocks of system~\cref{eq:principal}, since the
	last column blocks are zero for these rows, we have, after row and column
	permutations, 
	\begin{equation}
		\begin{bmatrix}
      C  & J_2^\top \\
			J_2 & A_{22}
		\end{bmatrix} \begin{bmatrix}
			\vb{\hat{x}}_{2} \\ \vb{\hat{x}}_{12}
		\end{bmatrix} = \begin{bmatrix}
			H^\top\vb{\hat{b}} \\ \vb{b}_2  
		\end{bmatrix}. \label{eq:lemma}
	\end{equation}
	(This system is the  updated $n \times n$ system of equations \cref{eq:mod}
	written in its  block $2 \times 2$ form.) Hence the vectors
	$\vb{\hat{x}}_{12}$ and $\vb{\hat{x}}_2$ are independent of the submatrix $J_1$.
\end{proof}

Note that the values of $\vb{\hat{x}}_2$ and $\vb{\hat{x}}_3$ are coupled
in \cref{eq:coupled}, i.e., we can express one in terms of the other.
Therefore, we need only one of them when solving the updated solution
$\vb{\hat{x}}$ in \cref{eq:mod}.

Since we have applied the permutation to the solution vector in
\cref{eq:principal}, we need to unpermute it to obtain the updated solution
$\vb{\hat{x}}$ in \cref{eq:mod}.
Hence we obtain 
\begin{equation}
	P\vb{\hat{x}} = \begin{bmatrix}
		\vb{\hat{x}}_{2} \\ \vb{\hat{x}}_{12}
	\end{bmatrix}.\label{eq:updated-solution}
\end{equation}

\subsection{Principal submatrix update}
We now extend the techniques described in the previous subsection to design
an augmented matrix approach to update the solution when $A$ is modified by a
principal submatrix update as in \cref{eq:update}. In this case, all the
changes are captured in the submatrix $C$ in \cref{eq:replace}, and we can
deduce that $C = H^\top \hat{A}H = H^\top AH - E$. Therefore, the submatrix
$J_2$ in \cref{eq:principal} remains unchanged from the original system and thus
$J_2 = A_{12}^\top$. As proven in the previous subsection, $J_1$ in
\cref{eq:principal} can be chosen arbitrarily. With the choice of $J_1 = A_{11}$,
we can show that 
\begin{equation}
	J = P^\top\begin{bmatrix}
		J_1 \\ J_2
	\end{bmatrix} = P^\top\begin{bmatrix}
		A_{11} \\ A_{12}^\top
	\end{bmatrix} = AH. 
\end{equation}
The last equation follows from 
\begin{equation}
A H = P^\top\begin{bmatrix}
		 		A_{11} &  A_{12} \\  A_{12}^\top & A_{22}
      \end{bmatrix} P H
		= P^\top\begin{bmatrix}
				A_{11} &  A_{12} \\  A_{12}^\top & A_{22} 
			\end{bmatrix}  \begin{bmatrix}  I_m \\ 0_{n-m} \end{bmatrix}
		= P^\top \begin{bmatrix} A_{11} \\ A_{12}^\top \end{bmatrix}.
\end{equation}
We can thus write \cref{eq:replace} as
\begin{equation}
	\begin{bmatrix}
		A & AH & H\\
		H^\top A & C & 0\\
		H^\top & 0 & 0
	\end{bmatrix} \begin{bmatrix}
		\vb{\hat{x}}_1\\
		\vb{\hat{x}}_2\\
		\vb{\hat{x}}_3
	\end{bmatrix} = \begin{bmatrix}
		\vb{b}\\
		H^\top\vb{\hat{b}}\\
		\vb{0}
	\end{bmatrix}.\label{eq:augEqn}
\end{equation}

Here is an example in which the principal submatrix at the $3$rd and
$5$th rows and columns are modified. The augmented system
\cref{eq:augEqn} would be
{\setlength{\arraycolsep}{3.75pt}%
\begin{equation}
	\left[\begin{array}{*{6}c|*{4}c}
			\multicolumn{6}{c|}{\multirow{6}{*}{$A$}} & \alpha_{13} & \alpha_{15} & 0 & 0\\
			\multicolumn{6}{c|}{} & \alpha_{23} & \alpha_{25} & 0 & 0\\
			\multicolumn{6}{c|}{} & \underline{\alpha_{33}} & \underline{\alpha_{35}} & 1 & 0\\
			\multicolumn{6}{c|}{} & \alpha_{43} & \alpha_{45} & 0 & 0\\
			\multicolumn{6}{c|}{} & \underline{\alpha_{53}} & \underline{\alpha_{55}} & 0 & 1\\
			\multicolumn{6}{c|}{} & \vdots & \vdots & \vdots & \vdots\\
			\hline 
			\alpha_{31} & \alpha_{32} & \underline{\alpha_{33}} & \alpha_{34} & \underline{\alpha_{35}} & \cdots & \hat{\alpha}_{33} & \hat{\alpha}_{35} & 0 & 0\\
			\alpha_{51} & \alpha_{52} & \underline{\alpha_{53}} & \alpha_{54} & \underline{\alpha_{55}} & \cdots & \hat{\alpha}_{53} & \hat{\alpha}_{55} & 0 & 0\\
			0 & 0 & 1 & 0 & 0 & \cdots & 0 & 0 & 0 & 0\\
			0 & 0 & 0 & 0 & 1 & \cdots & 0 & 0 & 0 & 0\\
	\end{array}\right] \hspace{-0.03in} \begin{bmatrix}
		\hat{\chi}_1 \\ \hat{\chi}_2 \\ \zeta_3 \\ \hat{\chi}_4 \\ \zeta_5 \\ \vdots \\ \hat{\chi}_3 \\ \hat{\chi}_5 \\ \delta_3 \\ \delta_5
	\end{bmatrix} \hspace{-0.015in} = \hspace{-0.015in} \begin{bmatrix}
		\beta_1 \\ \beta_2 \\ \beta_3 \\ \beta_4 \\ \beta_5 \\ \vdots \\ \hat{\beta}_3 \\ \hat{\beta}_5 \\ 0 \\ 0
	\end{bmatrix},
	\label{eq:augEg}
\end{equation}
}%
in which the $\zeta$ terms are constrained to be 0, the $\chi$ terms are the
permuted solutions to \cref{eq:mod}, and the $\delta$ terms are the
values of $\vb{\hat{x}}_3$ in \cref{eq:augEqn}.

\section{Solution method}
\label{sec:method}
In this section, we describe our algorithms to solve the system \cref{eq:replace}.
Suppose we have computed the LDL$^\top$ factorization of $A$ when solving the
original system $A \vb{x} = \vb{b}$. Here, $L$ is a unit lower triangular matrix and
$D$ is a diagonal matrix or block diagonal matrix with $1 \times 1$ or $2 \times 2$
blocks if $A$ is indefinite. A fill-reducing ordering and an ordering to maintain
numerical stability are usually used during the factorization, and thus a permuted
matrix of $A$ is factored, i.e.,  $\dot{P}^\top A \dot{P} = L D L^\top$ for some
permutation matrix $\dot{P}$. We assume that hereafter the permuted system
$\dot{P}^\top A\dot{P}\dot{P}^\top\vb{x} =\dot{P}^\top\vb{b}$ has replaced the
original system. Solutions to the original system $A \vb{x} = \vb{b}$ can then be
obtained by applying the inverse permutation $\dot{P}$.
For simplicity, we will not explicitly write the permutation matrices $\dot{P}$. We can
solve \cref{eq:augEqn} in two ways.

\paragraph{Iterative method}
With $A = LDL^\top$ as a block pivot, \cref{eq:augEqn} can be reduced to a
smaller system involving the symmetric matrix $S_1$, the Schur complement of
$A$, with a multiplication by $-1$:
\begin{equation}
	\underbrace{\begin{bmatrix}
		E & I\\
		I & H^\top A^{-1}H
	\end{bmatrix}}_{S_1} \begin{bmatrix}
		\vb{\hat{x}}_2\\
		\vb{\hat{x}}_3
	\end{bmatrix} = \begin{bmatrix}
		H^\top(\vb{b}-\vb{\hat{b}}) \\ H^\top A^{-1}\vb{b}
	\end{bmatrix},\label{eq:S1}
\end{equation}
where $E = H^\top A H - C$, which is the same $E$ as in \cref{eq:update}.
This can be shown by premultiplying and postmultiplying \cref{eq:update}
by $H^\top$ and $H$ respectively:
\begin{equation}
	C \equiv H^\top \hat{A} H = H^\top A H - H^\top H E H^\top H = H^\top A H - E.
\end{equation}
We can solve \cref{eq:S1} by an iterative method such as GMRES or MINRES.
Matrix-vector products with $S_1$ need a partial solve with $A = LDL^\top$
involving only the rows and columns selected by $H$ and $H^\top$, and products
with $E$, in each iteration. Note that $H^\top$ in the right-hand-side vector
selects the components from the difference vector $\vb{\hat{b}} - \vb{b}$ and the
solution $\vb{x}$ of the original system $A\vb{x} = \vb{b}$ if the changes in the
right-hand-side vector are only in the changed rows of $A$.

\paragraph{Direct method}
Alternatively, the $(2,1)$-block of $S_1$ can be used as a block pivot with
another Schur complement:
\begin{equation}
	\big(\underbrace{E H^\top A^{-1} H - I}_{S_2}\big)\vb{\hat{x}}_3 =
		E H^\top A^{-1}\vb{b} - H^\top(\vb{b}-\vb{\hat{b}}).\label{eq:x3}
\end{equation}
We can then solve this equation for $\vb{\hat{x}}_3$ using a direct solver
with an LU factorization of $S_2$, which can be constructed efficiently
as described later in \cref{sec:formS2}. Note that $S_2$ is unsymmetric
although both the augmented matrix in \cref{eq:augEqn} and $S_1$ in
\cref{eq:S1} are symmetric. This is because we have chosen an off-diagonal
block pivot in forming $S_2$, to avoid the computation of the inverse of
either $E$ or $H^\top A^{-1}H$ on the diagonal block of $S_1$.

\subsection{Solution to the modified system}\label{sec:sol}
It turns out that we only need to compute $\vb{\hat{x}}_3$ to obtain the 
full solution vector $\vb{\hat{x}}$ to the modified system \cref{eq:mod}. 
This can be done by making the following observation.
Premultiplying the first row block of \cref{eq:augEqn} by $A^{-1}$, and 
rearranging terms yields
\begin{equation}
	\vb{\hat{x}}_1 = A^{-1}\vb{b} - H\vb{\hat{x}}_2 - A^{-1} H\vb{\hat{x}}_3.
	\label{eq:row1}
\end{equation}
From \cref{eq:updated-solution}, we have
\begin{align}
	P\vb{\hat{x}} &= \begin{bmatrix}
		\vb{\hat{x}}_{2} \\ \vb{\hat{x}}_{12}
	\end{bmatrix}
	= \begin{bmatrix}
		\vb{\hat{x}}_{11} \\ \vb{\hat{x}}_{12}
	\end{bmatrix} + \begin{bmatrix}
		\vb{\hat{x}}_2 \\ \vb{0}
	\end{bmatrix}
	= P\vb{\hat{x}}_1 + PH\vb{\hat{x}}_2,
\end{align}
by using \cref{eq:PH,eq:Px1}, and the fact that $\vb{\hat{x}}_{11} = \vb{0}$.
Premultiplying both sides by $P^\top$ yields
\begin{equation}
	\vb{\hat{x}} = \vb{\hat{x}}_1 + H\vb{\hat{x}}_2.\label{eq:xhat}
\end{equation}
Substituting \cref{eq:row1} into \cref{eq:xhat}, we have
\begin{equation}
	\vb{\hat{x}} = A^{-1}\vb{b} - A^{-1} H\vb{\hat{x}}_3,\label{eq:sol}
\end{equation}
in which the first term is the solution to the original system.

\subsection{Relation to the Sherman-Morrison-Woodbury formula}
The solution $\hat{\vb{x}}$ in \cref{eq:mod} obtained by AMPS using the
direct approach can be expressed in a single equation by substituting
$\vb{\hat{x}}_3$ in \cref{eq:x3} into \cref{eq:sol}:
\begin{equation}
	\hat{\vb{x}} = A^{-1}\vb{b}-A^{-1}H\left(E H^\top A^{-1}H-I\right)^{-1}
	\left[E H^\top A^{-1}\vb{b}-H^\top(\vb{b}-\vb{\hat{b}})\right].
	\label{eq:ampssol1}
\end{equation}
In the case when the right-hand-side of \cref{eq:mod} does not change from
the original system, i.e. $\vb{\hat{b}} = \vb{b}$, \cref{eq:ampssol1} becomes
\begin{equation}
	\hat{\vb{x}} = \left[A^{-1}-A^{-1}H\left(E H^\top A^{-1}H-I\right)^{-1}
	E H^\top A^{-1}\right]\vb{b}.\label{eq:ampssol2}
\end{equation}
Using the Sherman-Morrison Woodbury formula
\begin{equation}
	\left(A+UCV\right)^{-1}=A^{-1}-A^{-1}U\left(C^{-1}+VA^{-1}U\right)^{-1}VA^{-1},
\end{equation}
the inverse of $\hat{A}$ in \cref{eq:update} can be expressed as
\begin{align}
	\hat{A}^{-1} &= \left[A + \left(H\right)\left(-I\right)\left(EH^\top\right)\right]^{-1}\nonumber\\
							 &= A^{-1} - A^{-1}H\left(-I + EH^\top A^{-1}H\right)^{-1}EH^\top A^{-1},
\end{align}
which when multiplied by $\vb{b}$ is identical to \cref{eq:ampssol2}.

\subsection{Forming the Schur complement \texorpdfstring{$S_2$}{S2} explicitly}
\label{sec:formS2}
In the previous subsection we have described the algorithm to solve
the modified system using the augmented formulation.
We now discuss how we form the matrix $W \equiv E H^\top A^{-1} H$ in
the Schur complement $S_2$ in \cref{eq:x3} by using partial triangular solves. 

From the factorization of $A$, the matrix $H^\top A^{-1} H$ can be expressed as
the product $H^\top L^{-\top} D^{-1} L^{-1} H$. Then $W^\top$ can be expressed as
\begin{equation}
	W^\top = H^\top L^{-\top}  D^{-1} L^{-1} H E.\label{eq:wtop}
\end{equation}
Recall that $E$ is symmetric. Let $X \equiv L^{-1} H E$. 
Premultiplying both sides of this equation by $L$, we have
\begin{equation}
	L X = H E \equiv \tilde{E}.\label{eq:lg}
\end{equation}
Observe that the right-hand-side of \cref{eq:lg} is a matrix $\tilde{E}$
mapping the $i^\text{th}$ row of $E$ to the $j_{i}^\text{th}$ row of
$\tilde{E}$ with the rest of $\tilde{E}$ filled with 0. For instance, if
the set of indices of updates is $\mathbb{S} = \{3, 5\}$, then \cref{eq:lg}
would be
\begin{equation}
	LX = HE = \begin{bmatrix}
		0 & 0\\
		0 & 0\\
		1 & 0\\
		0 & 0\\
		0 & 1\\
		\vdots & \vdots
	\end{bmatrix} E = \begin{bmatrix}
		0 & 0\\
		0 & 0\\\
		\epsilon_{11} & \epsilon_{12}\\
		0 & 0\\
		\epsilon_{21} & \epsilon_{22}\\
		\vdots & \vdots
	\end{bmatrix}.
\end{equation}
Since both $L$ and $\tilde{E}$ are sparse, we can use partial forward
substitution to solve for $X$.

Let $Y \equiv L^{-1} H$. We can again use partial forward substitution by
exploiting the sparsity in $H$ and $L$ to compute $Y$, as discussed in
the next subsection. Once we have the matrices $X$ and $Y$, we can compute
$W^\top$ as follows:
\begin{equation}
	W^\top = H^\top L^{-\top} D^{-1} X = Y^{\top} D^{-1} X.
\end{equation}

\subsubsection{Exploiting sparsity in the computations}
To describe how we exploit the sparsity in the matrices and the vectors 
to reduce the complexity of our algorithms,
we need a few concepts from sparse matrix theory as outlined below. 
We start with a few definitions to help with the discussions that follow.

\begin{definition}
	An $n \times n$ sparse matrix $A$ can be represented by a directed graph
	$G(A)$ whose vertices are the integers $1,\ldots,n$ and whose edges are
	\[
		\{(i,j)\vcentcolon i\ne j,\text{ and } \alpha_{ij}\ne 0\}.
	\]
	The edge $(i ,j)$ is directed from vertex $i$ to $j$.
	The set of edges is also called the \textit{(nonzero) structure} of $A$.
\end{definition}

The transitive reduction of a directed graph $G= (V, E)$ is obtained 
by deleting from the set of edges $E$ every edge  $(i,j)$ 
such that  there is a directed path from vertex $i$ to $j$ 
that does not use the edge $(i,j)$ itself.
  
\begin{definition}
	An elimination tree of a Cholesky factor $L$ is the transitive reduction of
	the directed graph $G(L)$.~\cite{liu}
\end{definition}

\begin{definition}
	The \textit{(nonzero)} structure of an $n$-vector $\vb{x}$ is
	\[
		\struct{x}\vcentcolon =\{i\vcentcolon \chi_i\ne 0\},
	\]
	which can be interpreted as a set of vertices of the directed graph of any
	$n \times n$ matrix. In this paper, for a vector $\vb{x}$,
	$\closure{A}{\vb{x}}$ refers to $\closure{A}{\struct{\vb{x}}}$.
\end{definition}

\begin{definition}
	Given a directed graph $G(A)$ and a subset of its vertices denoted by $V$, we
	say $V$ is closed with respect to $A$ if there is no edge of $G(A)$ that joins
	a vertex not in $V$ to a vertex in $V$; that is, $\nu_j\in V$ and
	$\alpha_{ij}\ne 0$ implies $\nu_i\in V$. The \textit{closure} of $V$ with
	respect to $A$ is the smallest closed set containing $V$,
	\[
		\closure{A}{V}\vcentcolon =
			\bigcap\{U\vcentcolon V\subseteq U,\text{ and } U\text{ is closed}\},
	\]
	which is the set of vertices of $G(A)$ from which there are directed paths in
	$G(A)$ to vertices in $V$.
\end{definition}

To compute the structure of $X$ in \cref{eq:lg}, we apply the following
theorem.
\begin{theorem}
	\label{thm1}
	Let the structures of $A$ and $\vb{b}$ be given. Whatever the values of the
	nonzeros in $A$ and $\vb{b}$, if $A$ is nonsingular then
	\[
		\struct{A^{-1}\vb{b}}\subseteq\closure{A}{\vb{b}}.
	\]
\end{theorem}

The proof of \cref{thm1} is due to Gilbert~\cite{gilbert}.
Hence the structure of each column of $X$ would be the closure of the 
nonzeros of the corresponding column of $\tilde{E}$ in the graph of $G(L)$.
Similarly, to compute the submatrix of $L^\top$ necessary to obtain the 
needed components of $\tilde{W}$, we can apply the following theorem.
\begin{theorem}
	\label{thm2}
	Suppose we need only some of the components of the solution vector $\vb{x}$
	of the system $A\vb{x} = \vb{b}$. Denote the needed components by
	$\vb{\tilde{x}}$. If $A$ is nonsingular, then the set of components in
	$\vb{b}$ needed is $\closure{A^\top}{\vb{\tilde{x}}}$.
\end{theorem}

The proof of \cref{thm2} can be found in \cite{yeung}.
Hence we can deduce that if $\mathbb{S}$ is the set of indices of updates,
the submatrix of $L^\top$ needed would also be the closure of $\mathbb{S}$,
which is the same as the row indices of nonzeros in the columns of $X$.
(Recall that $\mathbb{S}$ has cardinality $m$.)
Since $L$ is a Cholesky factor, this closure is equivalent to the union of
all the vertices on the paths from $\mathbb{S}$ to the root of the
elimination tree of $G(L)$, denoted as $P_{\mathbb{S}}$, as proven among
others in \cite{yeung}. We denote the size of this closure by $\rho$:
\begin{equation}
	\rho \equiv |\closure{L}{\mathbb{S}}| = \sum_{k\in P_{\mathbb{S}}}{|L_{*k}|}.
	\label{eq:rho}
\end{equation}
The upper bound on $\rho$ is the total number of nonzeros in $L$, denoted as
$|L|$. In practice since $m \ll n$, this upper bound is quite loose, and $\rho$ is
closer to a small constant times $m$ than to $|L|$.

\subsection{Complexity analysis}
The time complexity of principal submatrix updates using the symmetric
augmented formulation can be summarized in \cref{tab:summary}. Direct method
refers to the approach of solving for $\vb{\hat{x}}_3$ directly using
\cref{eq:x3}, and iterative method refers to the approach of applying an
iterative method to \cref{eq:S1}. Recall that $n$ is the size of the original
matrix $A$, $m$ is the size of the principal submatrix update $C$, while
$t$ denotes the number of iterations that the iterative method takes to
converge.
\begin{table*}[th]
	\small
	\setlength{\extrarowheight}{0.07in}
	\begin{tabularx}{\textwidth}{llp{1.7in}>{\centering\arraybackslash}X|>{\centering\arraybackslash}X}
		\toprule
		& & \multirow{2}{*}{Computation} & \multicolumn{2}{>{\centering\setlength\hsize{2\hsize}}X}{Complexity} \\
	 & & & Direct method & Iterative method \\ \hline
		\multicolumn{4}{l}{ Amortized initialization:}\\
		1 & \multicolumn{2}{l}{Compute LDL$^\top$ factorization of $A$} & \multicolumn{2}{c}{$O(n^{3/2})$ for planar networks}\\
		2 & \multicolumn{2}{l}{Compute $\vb{x} = A^{-1}\vb{b}$} & \multicolumn{2}{c}{$O(|L|)$}\\\hline
	\multicolumn{4}{l}{ Real-time update steps:} \\
	1 & \multicolumn{2}{l}{Obtain the submatrix $B$} & \multicolumn{2}{c}{$O(|B|) \leq O(m^2)$}\\
	2 & \multicolumn{2}{l}{Compute $E = H^\top A H - B$} & \multicolumn{2}{c}{$O(|E|) \leq O(m^2)$}\\
	3 & \multicolumn{2}{l}{Compute $W^\top= H^\top A^{-1} H E$} & $O(m\cdot\rho)$ & -\\
		& (a) & Form $\tilde{E} = HE$ & $O(|E|)$ & -\\
	 & (b) & Solve $L X = \tilde{E}$ & $O(m\cdot\rho)$ & -\\
	 & (c) & Solve $L^\top\tilde{W}^\top = D^{-1}X$ & $O(m\cdot\rho)$ & -\\
	 & (d) & Form $W^\top = H^\top \tilde{W}^\top$ & $O(m^2)$ & -\\
	4 & \multicolumn{2}{l}{Form $ W - I$} & $O(m)$ & -\\
	5 & \multicolumn{2}{l}{Form R.H.S. of \cref{eq:S1,eq:x3}} & $O(m)$ & $O(|E|+m)$\\
	6 & \multicolumn{2}{l}{Solve for $\vb{\hat{x}}_3$} & $O(m^3)$ & $O(t \cdot\rho)$\\
	7 & \multicolumn{2}{l}{Solve $\vb{\hat{x}} = \vb{x} - A^{-1} H \vb{\hat{x}}_3$} & \multicolumn{2}{c}{$O(|L|)$}\\
	\bottomrule
\end{tabularx}
\label{tab:summary}
\caption{Summary of time complexity}
\end{table*}

The overall time complexity of the direct method is dominated by either Step~3
(computing $W^\top$) or Step~7 (solving for $\vb{\hat{x}}$), i.e.,
$O(m\cdot\rho+|L|)$. For the iterative method, the  time complexity is dominated
by either Step~6 (solving for $\vb{\hat{x}}_3$) or Step~7 
(solving for $\vb{\hat{x}}$), i.e., $O(t\cdot\rho+|L|)$.
Hence the AMPS algorithms have the time complexities
\begin{equation}
	O(m\cdot\rho+|L|) \ (\text{direct})\quad\text{and}\quad
	O(t\cdot\rho+|L|) \  (\text{iterative}).
\end{equation}

In comparison, for CHOLMOD~\cite{davis}, the time complexity for updating the
Cholesky factor of the matrix, when row and column changes are made, is  
\begin{equation}
	O\left(\sum_{j\in\mathbb{S}}\left(\sum_{k\vcentcolon\overline{L}_{jk}\ne 0}|L_{*k}|
	+ \sum_{k\in\overline{P}_j}|\overline{L}_{*k}|\right)\right),\label{eq:cholmod}
\end{equation}
where $L$ is the original Cholesky factor, $\overline{L}$ is the modified
Cholesky factor and $\overline{P}_j$ is the path from node~$j$ to the root of
the elimination tree of $\overline{L}$. (Note that 
we have to add the cost $|\overline{L}|$ to compute the 
solution by solving the triangular system of equations.)

Consider the two inner sums in the expression for the complexity. The first
inner sum computes the total number of operations of Steps 1--4 in both
Algorithms~1 (Row Addition) and 2 (Row Deletion) in CHOLMOD. If we denote
$\mathbb{T}_j$ as the set of nodes~$k < j$ in $G(\hat{A})$ that have an edge
incident on node~$j$, then this sum is equivalent to the number of outgoing
edges of the closure of $\mathbb{T}_j$ in $G(\overline{L})$ up to node~$j$.
The second inner sum computes the number of operations needed for Step 5
(rank-1 update/downdate) in Algorithms~1 and 2 of CHOLMOD. This sum is
equivalent to the closure of~$\{j\}$ in the updated graph $G(\overline{L})$.
Combining the two summation terms, we can express the time complexity of
CHOLMOD in terms of the closures:
\begin{equation}
	O\left(\sum_{j\in\mathbb{S}}\closure{\overline{L}}{\mathbb{T}_j}\right)
	\le O\left(m\cdot\max_j\closure{\overline{L}}{\mathbb{T}_j}\right).
\end{equation}

We make two observations when comparing the AMPS algorithms with
\linebreak CHOLMOD. First, in general, the AMPS algorithms do not introduce
new fill-in elements in the Cholesky factor whereas fill-ins are possibly
introduced in CHOLMOD. However, this happens when the update introduces a
new nonzero entry in \linebreak row/column~$j$ of $\hat{A}$. In our
application to the contingency analysis for power flow, we only remove
connections between buses. Hence running CHOLMOD neither introduces fill-ins
to the factor nor changes the elimination tree. Second, since the nodes in
$\mathbb{T}_j$ are numbered less than $j$, the closure of $\mathbb{T}_j$ is
always larger than the closure of~$\{j\}$, whether or not the updated factor
$\overline{L}$ is different from $L$. In the case that row $j$ of $\overline{L}$
is relatively dense due to fill-in, the first inner sum in \cref{eq:cholmod} may be
the dominant term. On the other hand, the AMPS algorithms only need the closure 
from node~$j$ in $G(L)$.

\subsection{Comparison with other augmented methods}
Several algorithms have been proposed to solve a modified system of linear
equations using augmented matrices. Gill et al.~\cite{gill} used augmented
matrices and a factorization approach to update basis matrices in the
simplex algorithm for linear programming, motivated by the work of Bisschop and
Meeraus~\cite{bisschop1977,bisschop1980}. In their method, the matrix
 was
factored in a block-LU form as
\begin{equation}
	\begin{bmatrix}
		A & \hat{A}H\\
		H^\top & 
	\end{bmatrix} = \begin{bmatrix}
		L &\\
		\tilde{Z}^\top & \tilde{D}
	\end{bmatrix} \begin{bmatrix}
		U & \tilde{Y}\\
			& I
	\end{bmatrix}.\label{eq:gill}
\end{equation}
Here the matrix $L$ is unit-lower triangular and the matrix $U$ is upper triangular.
The matrices $\tilde{Y}$ and $\tilde{Z}$ are $n \times m$ submatrices of the block
factors and $\tilde{D}$ is the Schur complement of $A$.

Maes~\cite{maes} and Wong~\cite{wong} used a similar approach to implement
active-set QP solvers with symmetric augmentation to solve the Karush-Kuhn-Tucker
(KKT) matrices arising from Hessian updates and factored in a block-LU form as
\begin{equation}
	\begin{bmatrix}
		A & V\\
		V^\top & \tilde{C}
	\end{bmatrix} = \begin{bmatrix}
		L &\\
		Z^\top & I
	\end{bmatrix} \begin{bmatrix}
		U & Y\\
			& \tilde{S}
	\end{bmatrix}.\label{eq:kktLU}
\end{equation}
Here the submatrices $Y$ and $Z$ are $n \times 2m$ submatrices, doubling
the size of $\tilde{Y}$ and $\tilde{Z}$ in \cref{eq:gill}. These submatrices
were updated using sparse triangular solves and $\tilde{S}$ was updated
using a dense LU-type factorization. Comparing the augmented matrix in
\cref{eq:kktLU} with \cref{eq:replace}, we have
\begin{equation}
	V = \begin{bmatrix}
		AH & H
	\end{bmatrix}Q^\top\quad\text{and}\quad \tilde{C} = Q\begin{bmatrix}
		C & 0\\
		0 & 0
	\end{bmatrix}Q^\top,\label{eq:VC}
\end{equation}
for some permutation matrix $Q$.

To take advantage of symmetry, Maes and Wong factored the augmented matrix in
a symmetric block-LBL${}^\top$ form
\begin{equation}
	\begin{bmatrix}
		A & V\\
		V^\top & \tilde{C}
	\end{bmatrix} = \begin{bmatrix}
		L &\\
		Z^\top & I
	\end{bmatrix} \begin{bmatrix}
		D & \\
			&\tilde{D}
	\end{bmatrix} \begin{bmatrix}
		L^\top & Z\\
					 & I
	\end{bmatrix}.\label{eq:sc_factors}
\end{equation}
The major differences between the our methods and the KKT matrix
block-LU/block-LBL$^\top$ update method are as follows.
We exploited the explicit forms of the submatrices $Z$ and $\tilde{D}$
in the factorization when the update is a principal submatrix.
Specifically, if we factor the augmented matrix in \cref{eq:augEqn} as in
\cref{eq:sc_factors}, we have
\begin{equation}
	\begin{bmatrix}
		A & AH & H\\
		H^\top A & C & 0\\
		H^\top & 0 & 0
	\end{bmatrix} = \hat{L}\hat{D}\hat{L}^\top,\label{eq:augFactors}
\end{equation}
where
\begin{equation}
	\hat{L} =	\begin{bmatrix}
		L\\
		H^\top L & I\\
		H^\top L^{-\top} D^{-1} & 0 & I
	\end{bmatrix}\quad\text{and}\quad\hat{D} = \begin{bmatrix}
		D\\
		& -S_1
	\end{bmatrix}. \label{eq:factorHat}
\end{equation}
Here $\hat{L}$ is a lower triangular matrix, $\hat{D}$ is a matrix whose
$(1,1)$-block is (block) diagonal and the rest is the negation of the Schur
complement $S_1$ in \cref{eq:S1}. Combining the results from \cref{eq:VC}
and \cref{eq:factorHat}, we obtain the relationship between the factors in
\cref{eq:sc_factors} and those in our method:
\begin{equation}
	Z^\top = Q \begin{bmatrix}
		H^\top L\\
		H^\top L^{-\top} D^{-1}
	\end{bmatrix}\quad\text{and}\quad
	\tilde{D} = Q(-S_1)Q^\top.
\end{equation}
Hence, we do not need to construct $Z$ and $\tilde{D}$ as in the Gill et al.,
Maes, and Wong algorithms. We also make use of the structure of the factors in
computing the solution, whereas Maes updated the factors by treating the
augmentation submatrix $V$ as sparse and $\tilde{C}$ as dense. Finally, we
compute the solution to the modified system by explicitly using the solution to
the original system.

\section{Experimental results}
\label{sec:results}
The augmented matrix solution method was evaluated through a series of 
$N-k$ contingency analyses of two real-world power systems, 
the 3,120-bus Polish system from the  MATPOWER repository~\cite{zimmerman}
and the 14,090-bus WECC system; and a 777,646-bus generated system,
which is based on the \textit{case2736sp} system from the 
MATPOWER repository and the IEEE 123 bus distribution feeder~\cite{kersting}. 
The distribution feeder is balanced by equivalencing the load on 
each phase and extending the unbalanced laterals. 
Several distribution feeders are added at appropriate locations in 
the transmission case to create this system.

This section provides relevant implementation details and presents
experimental results, including comparisons with the PARDISO direct
solver~\cite{kuzmin, schenk2008, schenk2007} on the modified systems,
and the CHOLMOD direct solver that updates the factors of $A$ according to
$\hat{A}$.

Since the power flow systems follow Kirchhoff's current law, 
the admittance matrix $B$ in \cref{eq:DC} is a weighted Laplacian. 
A boundary condition is applied to fix the phase shift of a selected bus
called the slack bus, and the reduced system is nonsingular 
but with an eigenvalue close to zero. Hence in the power community a direct
solver is usually used to solve the system.

The estimated condition numbers of the admittance matrices $B$
calculated by using MATLAB's \texttt{condest} function are 
$1.2 \times 10^6$ for the 3,120-bus Polish system,
$2.1 \times 10^7$ for the 14,090-bus WECC system,  
and $9.9 \times 10^8$ for the 777,646-bus generated system. 
The estimated eigenvalues with the smallest magnitude calculated
by using MATLAB's \texttt{eigs} function are 
$5.0 \times 10^{-2}$ for the Polish system, 
$2.3 \times 10^{-3}$ for the WECC system,  
and $1.4 \times 10^{-4}$ for the generated system.

\begin{figure}[htbp]
  \centering
	\subfloat[$3,120$-bus Polish system]{
		\begin{tikzpicture}
			\begin{pgfinterruptboundingbox}
			\begin{semilogyaxis}[
					xlabel={number of edges removed},
					ylabel={time (s)},
					ymajorgrids=true,
					thick,
					legend cell align=left,
					legend style={font=\small,at={(1.01,0.5)},anchor=west,draw=none},
					legend entries={Aug. Direct,Aug. GMRES,CHOLMOD,PARDISO}
				]
				\addplot table [x=remove,y=direct_t] {case3120sp.dat};
				\addplot table [x=remove,y=gmres_t] {case3120sp.dat};
				\addplot table [x=remove,y=cholmod_t] {case3120sp.dat};
				\addplot table [x=remove,y=pardiso_t] {case3120sp.dat};
			\end{semilogyaxis}
			\end{pgfinterruptboundingbox}
			\useasboundingbox
				(current axis.below south west)
				rectangle (current axis.above north east) + (1,0);
		\end{tikzpicture}\label{fig:time_a}
	}\\
	\subfloat[$14,090$-bus WECC system]{
		\begin{tikzpicture}
			\begin{pgfinterruptboundingbox}
			\begin{semilogyaxis}[
					xlabel={number of edges removed},
					ylabel={time (s)},
					ymajorgrids=true,
					thick,
					legend cell align=left,
					legend style={font=\small,at={(1.01,0.5)},anchor=west,draw=none},
					legend entries={Aug. Direct,Aug. GMRES,CHOLMOD,PARDISO}
				]
				\addplot table [x=remove,y=direct_t] {casepti08.dat};
				\addplot table [x=remove,y=gmres_t] {casepti08.dat};
				\addplot table [x=remove,y=cholmod_t] {casepti08.dat};
				\addplot table [x=remove,y=pardiso_t] {casepti08.dat};
			\end{semilogyaxis}
			\end{pgfinterruptboundingbox}
			\useasboundingbox
				(current axis.below south west)
				rectangle (current axis.above north east) + (1,0);
		\end{tikzpicture}\label{fig:time_b}
	}\\
	\subfloat[$777,646$-bus generated system]{
		\begin{tikzpicture}
			\begin{pgfinterruptboundingbox}
			\begin{semilogyaxis}[
					xlabel={number of edges removed},
					ylabel={time (s)},
					ymajorgrids=true,
					ytickten={-1.2,-1.4,-1.6,-1.8},
					thick,
					legend cell align=left,
					legend style={font=\small,at={(1.01,0.5)},anchor=west,draw=none},
					legend entries={Aug. Direct,Aug. GMRES,CHOLMOD}
				]
				\addplot table [x=remove,y=direct_t] {case777646.dat}; 
				\addplot table [x=remove,y=gmres_t] {case777646.dat};
				\addplot table [x=remove,y=cholmod_t] {case777646.dat};
			\end{semilogyaxis}
			\end{pgfinterruptboundingbox}
			\useasboundingbox
				(current axis.below south west)
				rectangle (current axis.above north east) + (1,0);
		\end{tikzpicture}\label{fig:time_c}
	}
  \caption{Timing results of compared methods}
  \label{fig:time}
\end{figure}
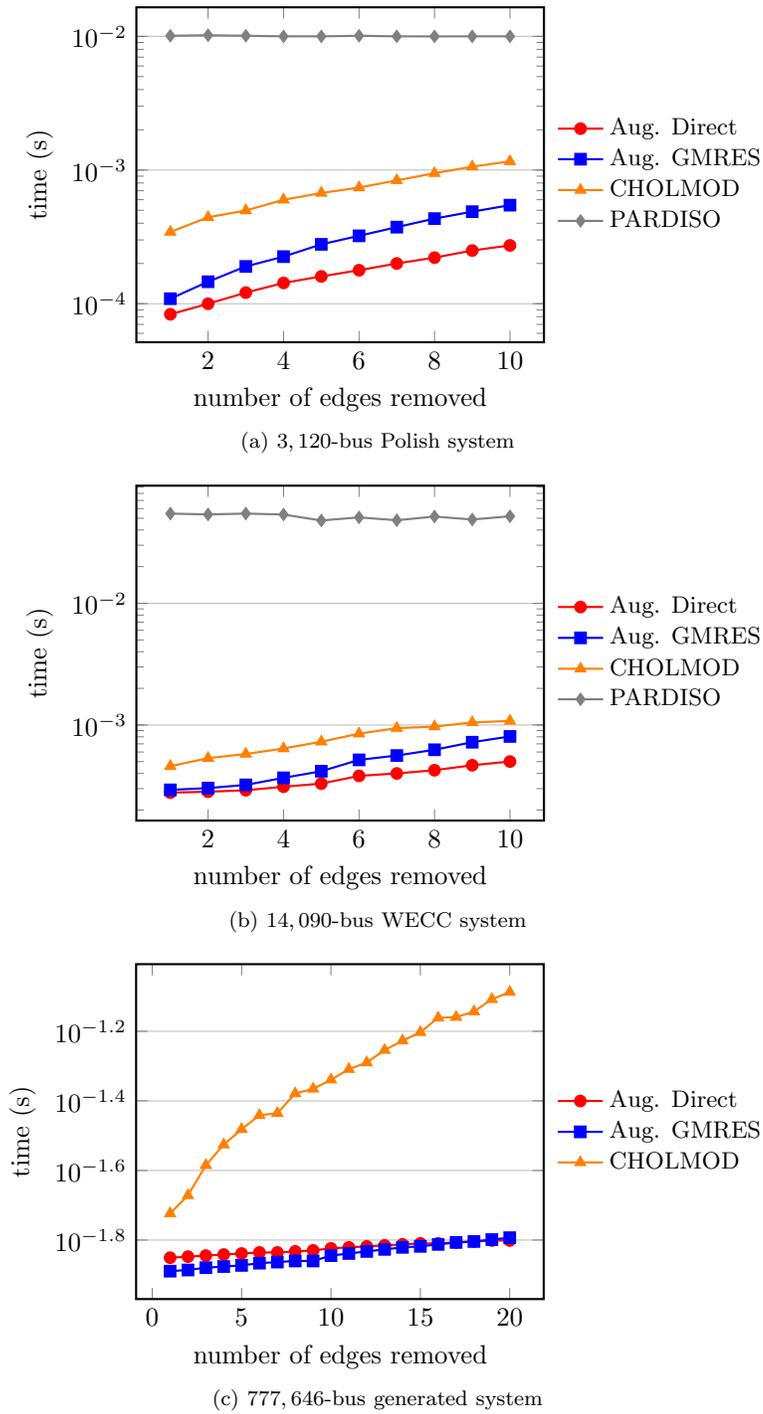
\begin{figure}[htbp]
	\centering
	\begin{tikzpicture}
		\begin{pgfinterruptboundingbox}
			\begin{axis}[stack plots=y,
					area cycle list,
					area legend,
					ymajorgrids=true,
					ymax=9e-2,
					enlarge y limits=false,
					ylabel={time (s)},
					xlabel={number of edges removed},
					thick,
					legend cell align=left,
					legend style={font=\small,at={(1.01,0.5)},anchor=west,draw=none},
					legend entries={Factor Update,Solution}
				]
				\pgfplotsset{cycle list shift=1}
				\addplot table [x=remove,y=cholmod_fac] {case777646.dat} \closedcycle;
				\pgfplotsset{cycle list shift=2}
				\addplot table [x=remove,y=cholmod_sol] {case777646.dat} \closedcycle;
			\end{axis}
		\end{pgfinterruptboundingbox}
		\useasboundingbox
		(current axis.below south west)
		rectangle (current axis.above north east) + (1,0);
	\end{tikzpicture}
  \caption{Breakdown of the time of the CHOLMOD method for the $777,646$-bus generated system}
	\label{fig:cholmod_bd}
\end{figure}
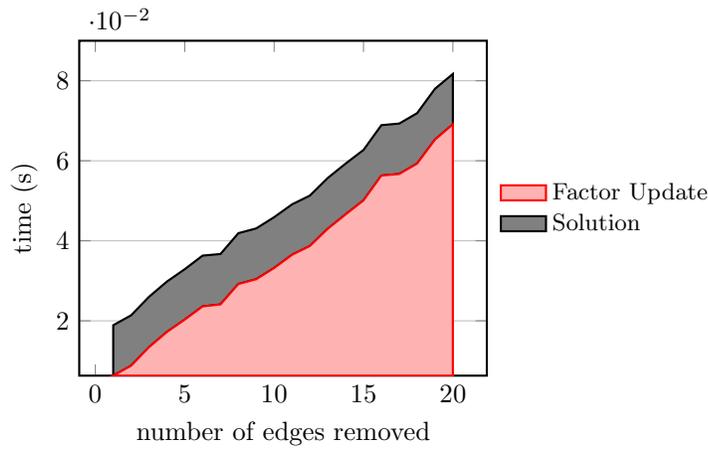
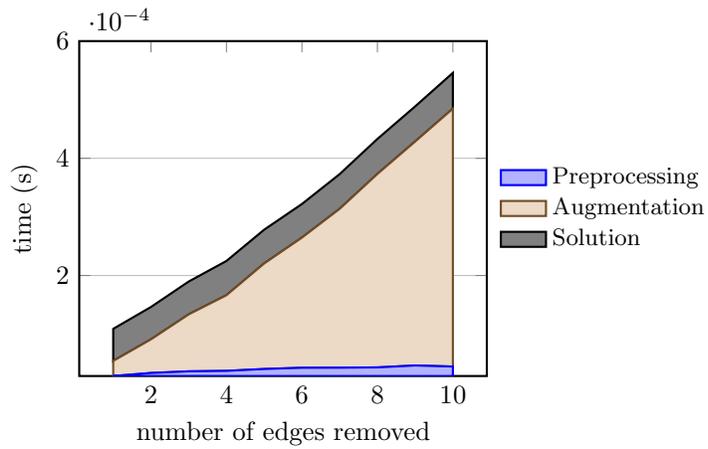
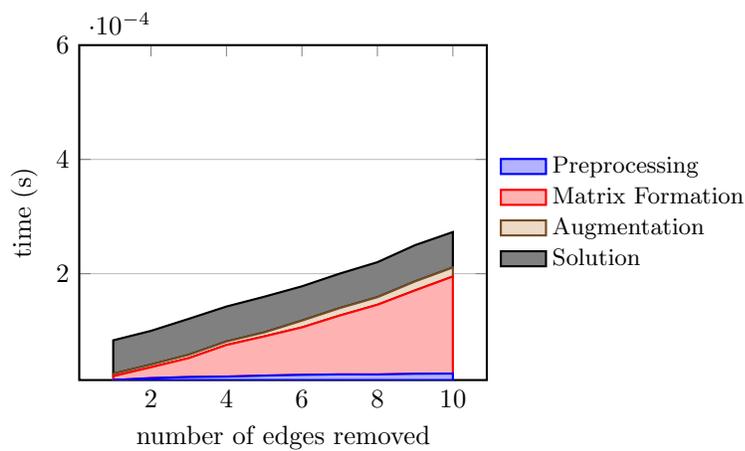
\begin{figure}[htbp]
  \centering
	\subfloat[Iterative method]{
		\begin{tikzpicture}
			\begin{pgfinterruptboundingbox}
			\begin{axis}[stack plots=y,
					area cycle list,
					area legend,
					ymajorgrids=true,
					ymax=6e-4,
					enlarge y limits=false,
					ylabel={time (s)},
					xlabel={number of edges removed},
					thick,
					legend cell align=left,
					legend style={font=\small,at={(1.01,0.5)},anchor=west,draw=none},
					legend entries={Preprocessing,Augmentation,Solution}
				]
				\addplot table [x=remove,y=gmres_pp] {case3120sp.dat} \closedcycle;
				\pgfplotsset{cycle list shift=1}
				\addplot table [x=remove,y=gmres_aug] {case3120sp.dat} \closedcycle;
				\addplot table [x=remove,y=gmres_sol] {case3120sp.dat} \closedcycle;
			\end{axis}
			\end{pgfinterruptboundingbox}
			\useasboundingbox
				(current axis.below south west)
				rectangle (current axis.above north east) + (1,0);
		\end{tikzpicture}\label{fig:3120bd_a}
	}\\
	\subfloat[Direct method]{
		\begin{tikzpicture}
			\begin{pgfinterruptboundingbox}
			\begin{axis}[stack plots=y,
					area cycle list,
					area legend,
					ymajorgrids=true,
					ymax=6e-4,
					enlarge y limits=false,
					ylabel={time (s)},
					xlabel={number of edges removed},
					thick,
					legend cell align=left,
					legend style={font=\small,at={(1.01,0.5)},anchor=west,draw=none},
					legend entries={Preprocessing,Matrix Formation,Augmentation,Solution}
				]
				\addplot table [x=remove,y=direct_pp] {case3120sp.dat} \closedcycle;
				\addplot table [x=remove,y=direct_mat] {case3120sp.dat} \closedcycle;
				\addplot table [x=remove,y=direct_aug] {case3120sp.dat} \closedcycle;
				\addplot table [x=remove,y=direct_sol] {case3120sp.dat} \closedcycle;
			\end{axis}
			\end{pgfinterruptboundingbox}
			\useasboundingbox
				(current axis.below south west)
				rectangle (current axis.above north east) + (1,0);
		\end{tikzpicture}\label{fig:3120bd_b}
	}
	\caption{Breakdown of the time for the $3,120$-bus Polish system}
  \label{fig:3120bd}
\end{figure}
\begin{figure}[htbp]
  \centering
	\subfloat[Iterative method]{
		\begin{tikzpicture}
			\begin{pgfinterruptboundingbox}
			\begin{axis}[stack plots=y,
					area cycle list,
					area legend,
					ymajorgrids=true,
					ymax=1.7e-2,
					enlarge y limits=false,
					ylabel={time (s)},
					xlabel={number of edges removed},
					thick,
					legend cell align=left,
					legend style={font=\small,at={(1.01,0.5)},anchor=west,draw=none},
					legend entries={Preprocessing,Augmentation,Solution}
				]
				\addplot table [x=remove,y=gmres_pp] {case777646.dat} \closedcycle;
				\pgfplotsset{cycle list shift=1}
				\addplot table [x=remove,y=gmres_aug] {case777646.dat} \closedcycle;
				\addplot table [x=remove,y=gmres_sol] {case777646.dat} \closedcycle;
			\end{axis}
			\end{pgfinterruptboundingbox}
			\useasboundingbox
				(current axis.below south west)
				rectangle (current axis.above north east) + (1,0);
		\end{tikzpicture}\label{fig:777646bd_a}
	}\\
	\subfloat[Direct method]{
		\begin{tikzpicture}
			\begin{pgfinterruptboundingbox}
			\begin{axis}[stack plots=y,
					area cycle list,
					area legend,
					ymajorgrids=true,
					ymax=1.7e-2,
					enlarge y limits=false,
					ylabel={time (s)},
					xlabel={number of edges removed},
					thick,
					legend cell align=left,
					legend style={font=\small,at={(1.01,0.5)},anchor=west,draw=none},
					legend entries={Preprocessing,Matrix Formation,Augmentation,Solution}
				]
				\addplot table [x=remove,y=direct_pp] {case777646.dat} \closedcycle;
				\addplot table [x=remove,y=direct_mat] {case777646.dat} \closedcycle;
				\addplot table [x=remove,y=direct_aug] {case777646.dat} \closedcycle;
				\addplot table [x=remove,y=direct_sol] {case777646.dat} \closedcycle;
			\end{axis}
			\end{pgfinterruptboundingbox}
			\useasboundingbox
				(current axis.below south west)
				rectangle (current axis.above north east) + (1,0);
		\end{tikzpicture}\label{fig:777646bd_b}
	}
	\caption{Breakdown of the time for the $777,646$-bus generated system}
  \label{fig:777646bd}
\end{figure}
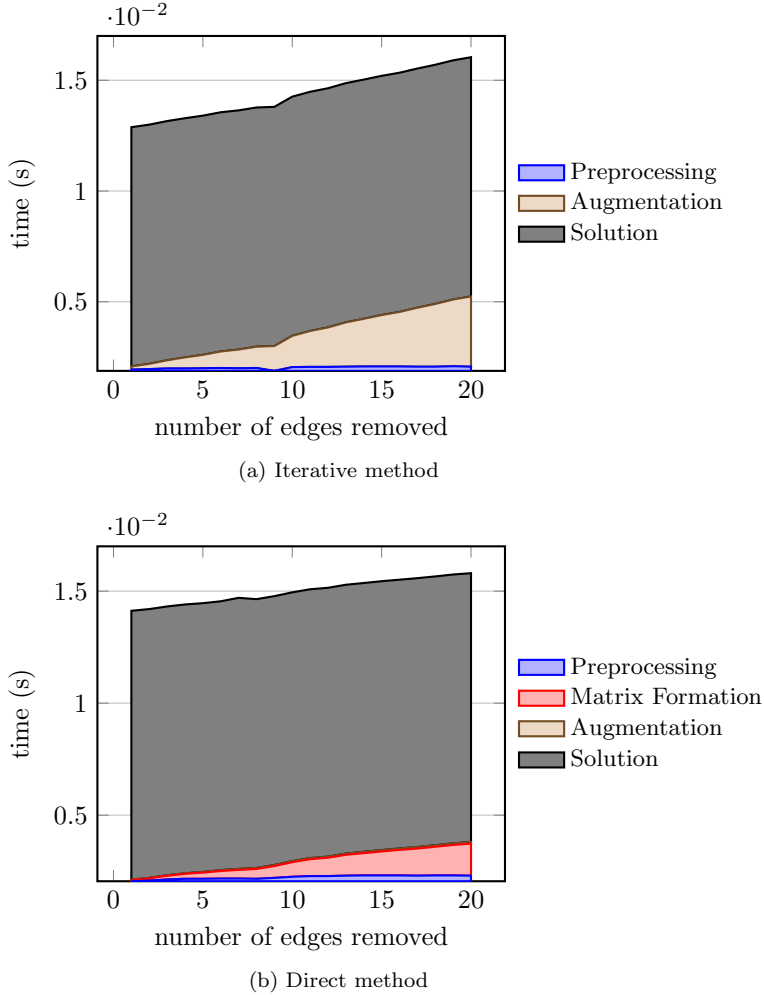

\subsection{Implementation}
All experiments were conducted on a desktop computer with four 8-core 
Intel Xeon E5-2670 processors running at 2.6 GHz with 20 GB cache and 256 GB RAM.
All reported times represent the average of 20 runs.

The precomputed LDL$^\top$ factorizations of the admittance matrices were 
computed using Oblio, a direct solver library for solving sparse symmetric
linear systems of equations with data structure support for dynamic pivoting
using $1 \times 1$ and $2 \times 2$ pivots~\cite{dobrian}. Both the GMRES
iterative solver used in \cref{eq:S1} and the PARDISO solver applied to
\cref{eq:mod} for comparison purposes were from the Intel Math Kernel Library
(MKL)~\cite{mkl}. The CHOLMOD solver applied to \cref{eq:mod} was from the
SparseSuite package. The remainder of the code was written by the authors.

All matrices were stored in sparse matrix format to reduce both the storage
space and access time. 

\begin{table}[htbp]
	\begin{tabularx}{\textwidth}{m{0.89in}XXXX}
		\toprule
		Problem & Aug. Direct & Aug. GMRES & PARDISO & CHOLMOD\\\midrule
		$3,120$-bus Polish system & $2 \times 10^{-13}$ & $3 \times 10^{-13}$ & $2 \times 10^{-13}$ & $2 \times 10^{-13}$\\
		$14,070$-bus WECC system  & $4 \times 10^{-13}$ & $4 \times 10^{-13}$ & $4 \times 10^{-13}$ & $5 \times 10^{-13}$\\
		$777,646$-bus system  & $6 \times 10^{-12}$ & $6 \times 10^{-12}$ & $6 \times 10^{-12}$ & $5 \times 10^{-12}$\\\bottomrule
	\end{tabularx}
	\label{tab:error}
	\caption{Average relative residual norms $(\big\|\hat{A}\vb{\hat{x}}-\vb{\hat{b}}\big\|_2/\big\|\vb{\hat{b}}\big\|_2)$ for each problem}
\end{table}

\subsection{Experiments}
In our $N - k$ contingency analysis experiments we remove $k$ out of $N$
connections in the power grid and form the modified system 
\cref{eq:mod}. This corresponds to a principal submatrix update as
described in \cref{eq:update}, where $H$ is formed by the columns of the
identity matrix corresponding to the end-points of the removed connections,
and $m \le 2k$.

We compare the solution of the augmented system using an iterative solver 
on \cref{eq:S1} and using \cref{eq:x3} by means of the LU factorization
of the Schur complement matrix $S_2$. 
Note that since $B$ in \cref{eq:DC} is a weighted Laplacian, the update matrix
$E$ at the $(1,1)$-block of \cref{eq:S1} is singular,
and thus the whole matrix is symmetric indefinite.
We have used the MINRES method and the generalized minimum residual (GMRES) 
method to solve these indefinite systems. 
For MINRES, the average solve time per iteration is faster than
the GMRES method, but since it converged slowly and needed more iterations
than GMRES, the total solve time was higher than the latter. 
Hence we report times obtained from GMRES. 
We also compare our augmented system with PARDISO and CHOLMOD being applied
to \cref{eq:mod}.

The LDL$^\top$ factorization times using Oblio were $0.0306$ seconds 
for the $3,120$-bus Polish system,
$0.156$ seconds for the $14,070$-bus WECC system, and $2.53$ seconds for 
the $777,646$-bus generated system.
In comparison, the average factorization times using PARDISO were 
$0.00735$ seconds for the $3,120$-bus Polish system,
$0.039$ seconds for the $14,070$-bus WECC system, 
and $2.37$ seconds for the $777,646$-bus generated system.
Although Oblio did not perform as well as PARDISO on the smaller problems, 
it provides the ability to extract the factors, which is essential for
closure computation and the sparsity-exploiting triangular solves.

In \cref{fig:time}, we plot the time to compute the updated solution 
when up to $20$ edges are removed from the grid. 
The augmented methods outperform
both PARDISO and CHOLMOD on all three power grids. 
The time taken by  PARDISO for the $777,646$-bus generated system is not 
plotted in \cref{fig:time_c} to better differentiate  
the relative performance of  our methods with CHOLMOD. 
For this large grid, PARDISO took approximately $2.4$ seconds for solving
each modified system, which is two orders of magnitude 
($149-186$ times) slower than our augmented iterative method.
In comparison, CHOLMOD computed the solutions $1.47-5.09$ times slower
than our augmented iterative method.
We also observe that the augmented  methods scale much better than CHOLMOD 
as the number of edges removed (size of the updates) increases.
However, the number of fill-ins, if any, introduced by CHOLMOD
is insignificant, as we can see from \cref{fig:cholmod_bd} that only
the factor update time increases when the number of edges removed increases.

\Cref{fig:3120bd,fig:777646bd} show the breakdown of the total time 
used in solving the updated systems using our augmented direct 
and iterative methods on the $3,120$-bus Polish
system and the $777,646$-bus generated system. 
Here \textit{Preprocessing} is the step of computing the
closure of the modified rows and columns in the graph of $G(L)$, 
and extracting the necessary submatrix of $L$ for solving for 
$\vb{\hat{x}}_3$ in \cref{eq:S1,eq:x3}. 
\textit{Augmentation} refers to the step  of solving 
\cref{eq:S1} for the iterative method 
and \cref{eq:x3} for direct method. 
\textit{Matrix Formation} corresponds to the step of forming $W$ as 
described in \Cref{sec:formS2}. 
\textit{Solution} is the step of computing  the solution to 
the modified system in \cref{eq:sol}.

It can be seen that for a small system the time is
dominated by the augmentation part in \cref{eq:S1}
(Step 6 in \cref{tab:summary}) for the iterative method,
or by the matrix formation of the reduced system in 
\cref{eq:x3} (Step 3 in \cref{tab:summary}) in the direct method. 
Hence the product of $m$ or $t$ (the number of steps of the iterative solver)
with $O(\rho)$ is the dominant term. 
On the other hand, for a large system the time is dominated 
by the computation of $\vb{\hat{x}}$ in \cref{eq:sol} 
(Step 7 in \cref{tab:summary}), which has $O(|L|)$ time complexity. 

The experimental results also indicate that the augmented solution methods do
not lead to difficulties with solution accuracy. \cref{tab:error} summarizes the
average relative residual norms for the solutions computed by our augmented methods,
PARDISO and CHOLMOD.

\section{Conclusions and future work}
\label{sec:conclusions}
We have formulated two algorithms using an augmented matrix approach to solve
linear systems of equations when the system is updated by a principal
submatrix. The algorithms use either a direct method or a hybrid of direct and
iterative methods. We applied the algorithms to assess the security of power
grids, and demonstrated that we could do $N - k$ contingency analysis by 
removing $k = 20$ connections in a grid with $778,000$  buses in
about 16 milliseconds.
The augmented solution methods have been experimentally shown to offer
advantages in both speed and reliability, relative to a 
direct solver (two orders of magnitude faster),
or a solver that updates the Cholesky factors ($1.5$ to $5$ times faster),
and scales better with an increasing $k$, the number of connections removed.
We believe that our algorithms are able to solve 
much larger dynamic security analysis problems in the power grid 
than previous work. 
 
In the future, we plan to extend our augmented solution method to problems
where the updated system of equations has a different size than the original
system, as in finite element applications~\cite{yeung}.


\section*{Acknowledgments}
We thank Dr.\ Jessica Crouch for collaborating with us 
in our earlier work on the augmented matrix approach to update solutions of 
linear systems of equations for visualizing and simulating surgery.
We also thank Dr.\ Mallikarjuna Vallem for providing us
the $777,646$-bus generated system for our experiments on contingency
analysis of power flow systems.
We are grateful to two anonymous referees for their helpful 
comments, which have improved the presentation of our manuscript.

\bibliographystyle{siamplain}
\bibliography{references}
\end{document}